\documentclass{llncs}
\usepackage[applemac]{inputenc}
\usepackage[english]{babel}
\usepackage{latexsym}
\usepackage{amssymb}
\usepackage{amsmath}
\usepackage{ifthen,xspace}
\usepackage{tikz}

\usetikzlibrary{snakes}
\usepackage{enumerate}
\newcommand{\vdd}[1]{
}

\newcommand{\anca}[1]{
}

\newenvironment{whoevercomesin}{\noindent\color{green} New guru : }{}

\usepackage{theorem}
\usepackage{prettyref}

\newrefformat{thm}{Thm.~\ref{#1}}
\newrefformat{lem}{Lem.~\ref{#1}}
\newrefformat{dfn}{Def.~\ref{#1}}
\newrefformat{cor}{Cor.~\ref{#1}} 
\newrefformat{prop}{Prop.~\ref{#1}}
\newrefformat{sec}{Sec.~\ref{#1}}
\newrefformat{ex}{Ex.~\ref{#1}}
\newrefformat{fig}{Fig.~\ref{#1}}
\newcommand{\refthm}[1]{Thm.~\ref{#1}}

\newcommand{\reflem}[1]{Lem.~\ref{#1}}

\newcommand{\Reach}{\mathop{\mathrm{Reach}}}

\newcommand{\PSPACE}{\ensuremath{\mathrm{PSPACE}}\xspace}

\newcommand{\set}[2]{\{#1 \mid #2\}}
\newcommand{\oneset}[1]{\{\mathinner{#1}\}}
\newcommand{\os}[1]{\{\mathinner{#1}\}}
\newcommand{\smallset}[1]{\{\mathinner{#1}\}}

\newcommand{\abs}[1]{\left|\mathinner{#1}\right|}

\newcommand{\N}{\mathbb{N}}

\newcommand{\R}{\mathbb{R}} 

\renewcommand{\P}{\mathbb{P}}

\newcommand{ \ov}[1]{ \overline{#1}\, }
\newcommand{\IFF}{if and only if\xspace}

\newcommand{\sse}{\subseteq}
\newcommand{\es}{\emptyset}
\newcommand{\sm}{\setminus}

\renewcommand{\phi}{\varphi}

\renewcommand{\SS}{\Sigma}
\newcommand{\GG}{\Gamma}
\newcommand{\oo}{\omega}
\newcommand{\alp}{\alpha}

\newcommand{\del}{\delta}

\newcommand{\Sig}{\Sigma}
\newcommand{\Gam}{\Gamma}

\newcommand{\alf}{\mathop{\mathrm{alph}}}
\newcommand{\alfinf}{\mathop{\mathrm{alphinf}}}

\newcommand{\Aa}{\mathcal{A}}
\newcommand{\Bb}{\mathcal{B}}
\newcommand{\Cc}{\mathcal{C}}
\newcommand{\Ll}{\mathcal{L}}
\newcommand{\Ff}{\mathcal{F}}
\newcommand{\Mm}{\mathcal{M}}
\newcommand{\Lin}{\textit{Lin}}
\newcommand{\dom}{\textit{dom}}
\newcommand{\tuple}[1]{\langle #1 \rangle}
\newcommand{\de}{\delta}
\newcommand{\D}{\Delta}
\newcommand{\Proc}{\textit{Proc}}
\newcommand{\act}[1]{\stackrel{#1}{\longrightarrow}}

\newcommand{\pref}{\textit{pref}}
\newcommand{\bean}{\begin{eqnarray*}}
\newcommand{\eean}{\end{eqnarray*}}

\newcommand{\locmon}{locally monitorable\xspace}
\newcommand{\moni}{monitorable\xspace}

\newcommand{\PGd}{\mathrm{PG}_{\delta}}
\newcommand{\PFs}{\mathrm{PF}_{\sigma}}

\begin{document}
\title{On distributed monitoring of asynchronous systems}

\author{Volker Diekert\inst{1} \and Anca Muscholl\inst{2}}
\institute{Universit\"at Stuttgart, FMI, Germany 
 \and LaBRI, Univ. of Bordeaux, France
}

\maketitle

\section{Introduction}

Distributed systems are notoriously difficult to understand and
analyze in order to assert their correction w.r.t.~given
properties. They often exhibit a huge number of different behaviors,
as soon as the active entities (peers, agents, processes, \ldots)
behave in an asynchronous manner. Already the modelization of such
systems is a non-trivial task, let alone their formal verification.

Several automata-based distributed models 
have been proposed and studied over the past twenty years,
capturing various aspects of distributed behavior.  Depending on
the motivation, such models fall into two large categories. In the first
one we find rather simple models, expressing basic synchronization
mechanisms, like Petri nets or communicating automata. In the
second category we see more sophisticated models, conceived
for supporting practical system design, like statecharts or
I/O automata. It is clear that being able to develop
automated verification techniques requires a good understanding of the
simpler models, in particular since more complex ones are often built
as a combination of  basic models. 

The purpose of this paper is to discuss the problem of distributed
monitoring on a simple model of finite-state distributed automata
based on shared actions, called \emph{asynchronous automata}. Monitoring
is a question related to runtime verification: assume that we have
to check a property $L$ against an unknown or very complex system
$\Aa$, so that classical static analysis is not possible. Therefore
instead of model-checking a~\emph{monitor} is used, that checks the
property on the underlying system at runtime. The question is which
properties can be checked in this way, that is, which properties  $L$
are \emph{monitorable}. A classical
example for monitorable properties are safety properties, like ``no alarm is raised''. A monitor
for a property $L$ is an automaton $\Mm_L$ that after each finite
execution tells whether (1) every possible extension of the execution
is in $L$, or (2) every possible extension is in the complement 
of $L$, or neither (1) nor (2)
holds. % Note that the simplest variant of monitors
% $M_\phi$ use no knowledge about $S$. It is however reasonable to define
% monitorability also w.r.t.~some partial knowledge about $S$.
The notion of monitorable  properties has been proposed by Pnueli and
Zaks~\cite{PnueliZ06}, and the theory has been extended to various kinds of
systems, for instance to probabilistic systems~\cite{ChadhaSV09jacm,GPS09}
or real-time systems~\cite{BauerLS06b,LeuckerBS10Tosem}.
 
We are interested here in monitoring distributed systems modelled as
asynchronous automata. It is natural to require that monitors  should be of the same
kind as the underlying system, so we consider here distributed monitoring. A distributed monitor does not have a global view of the system,
therefore we propose the notion of \emph{\locmon} trace language.  Our main
result shows that if the distributed alphabet of actions is connected and if $L$
is a set of $\GG$-infinite traces (for some subset of processes $\GG$)
such that both $L$ and its complement $L^c$
are countable unions of locally safety languages, then $L$ is
\locmon. We also show that over $\GG$-infinite traces, recognizable countable
unions of locally safety languages are precisely the complements of 
deterministic languages.

% The second problem we discuss here comes from control theory and describes the
% situation where actions are controlled either by the system or by
% the environment. Given a property $\phi$ and an automaton $\Aa$, the
% goal is to find a controller $\Cc$ such that in the synchronized
% product $\Aa \times \Cc$, every execution satisfies $\phi$. The
% controller cannot forbid uncontrollable actions, so this setting
% corresponds to a game. When $\Aa$ is a Zielonka automaton, the
% controller $\Cc$ is of the same type, which means that we look for
% controllers that work on the same distributed architecture, i.e., they
% are not allowed to introduce additional synchronizations. The
% decidability of this particular class of distributed games is a challenging
% open problem. We describe here the control problem and its relation
% with other variants of distributed games, and report on some recent
% progress. 

\section{Preliminaries}\label{sec:prel} 

The idea of describing concurrency by a fixed independence relation on
a given set of actions $\SS$ goes back to the late seventies, to
Mazurkiewicz~\cite{maz77} and Keller~\cite{kel73} (see
also~\cite{DieRoz95}). One can start with a~\emph{distributed action
  alphabet} $(\SS,\dom)$ on a finite set $\Proc$ of processes, where
$\dom:\SS \to (2^{\Proc}\setminus \es)$ is a \emph{location
  function}. The location $\dom(a)$ of action $a \in\SS$ comprises all
processes that need to synchronize in order to perform this action. It
defines in a natural way an~\emph{independence relation} $I \subseteq
\SS\times\SS$ by letting $(a,b) \in I$ \IFF $\dom(a) \cap \dom(b)
=\es$. 

The execution order of two independent actions $(a,b)\in I$ is
irrelevant, they can be executed as $a,b$, or $b,a$ - or even
concurrently. More generally, we can consider the congruence $\sim_I$
on $\SS^*$ generated by $I$. An equivalence class $[w]_I$ of $\sim_I$
is called a (finite) Mazurkiewicz \emph{trace}, and it can be also
viewed as labeled pomset $t=\tuple{V,\le,\lambda}$ of a special kind: if
$w=a_0 \cdots a_n$ then the vertex set is $V=\oneset{0,\ldots,n}$, the
labeling function is $\lambda(i)=a_i$ and
$\le \mathop{=} (\set{(i,j)}{i<j, (a_i,a_j) \notin I})^*$ is the
partial order. The word
$w$ is a~\emph{linearization} of $t$ defined as
above, i.e., a total order compatible with the partial order of $t$. 

\emph{Infinite traces} can be defined is a similar way from
$\oo$-words. Finite and infinite traces are also called~\emph{real
  traces}, and the set of real traces is written
$\R(\SS,I)$ (or simply $\R$ when $\SS,I$ are clear from the
context). A trace $t$ is a \emph{prefix} of a trace $t'$ (denotes as $t \le t'$)
if $t$ is isomorphic to a downwards-closed subset of $t'$. The set of
prefixes of $t$ is denoted $\pref(t)$.
If $L \subseteq \R$ then we denote by $\Lin(L) \subseteq
\SS^\infty$ the set of linearizations of traces from $L$.

A language $K \subseteq \SS^\infty$ is called \emph{trace-closed} if
$K=\Lin(L)$ for some $L \subseteq \R$. Whenever convenient, we talk
about trace languages $L \subseteq \R$ or trace-closed word languages
$K \subseteq \SS^\infty$ in equivalent terms. A language $L \subseteq \R$ is
\emph{recognizable} if $\Lin(L) \subseteq \SS^\infty$ is a regular
language of finite and infinite words.

Linear temporal properties like \emph{safety} and
\emph{liveness}~\cite{Pnu77}  can
be translated into topological properties, as
closed and dense sets in the Cantor topology. For real traces, these
notions generalize smoothly to the Scott topology, by replacing word
prefixes by trace prefixes. The Scott topology corresponds to a global
view in traces, where one needs to reason on global configurations,
i.e., configurations involving several processes. However, in the
setting of monitoring  that we discuss here, such a global
view is not available. Therefore we use here \emph{local safety} as
basic notion, as introduced in~\cite{dg09thiagu} and explained in the
following.

A trace $t=\tuple{V,\le,\lambda}$ is called~\emph{prime} if it is
finite and has a unique maximal element. That is, $|\max(t)|=1$, where
$\max(t)$ is the set of maximal elements of $t$ w.r.t.~the partial
order $\le$. The set of prime traces in
$\R$ is denoted $\P(\R)$. The set of prime prefixes of elements of $L \subseteq
\R$ is denoted $\P(L)$.

\begin{definition} Let $L \subseteq \R$.
  \begin{enumerate}
  \item $L$ is called \emph{prime-open} if it is of the form
  $\bigcup\set{p\R}{p \in U}$ for some $U \sse \P$. Complements of
  prime-open sets are called \emph{prime-closed}. 
\item $\ov L$ is the intersection of all prime-closed sets containing
  $L$ (and denoted as prime-closure of $L$).  Note that $\ov L$ is
  prime-closed.  
\item A prime-closed, recognizable language $L \subseteq \R$ is called
  a~\emph{locally safety language}.
  \end{enumerate}
  \end{definition}

\begin{remark}
  \begin{enumerate}
  \item  Every prime-open set is also Scott-open, and prime-open sets are
  closed under union, but not under intersection. As an example
  consider $a\R \cap
  b\R$ which is not prime-open for $(a,b)\in I$.
\item A first-order locally safety language $L \subseteq \R$ is 
 a prime-closed set such that $\Lin(L)$ is a first-order language. It
 is known from~\cite{dg09thiagu} that first-order locally safety languages are
 characterized by formulas of the form $G \, \psi$, with $\psi$ a past
 formula in a local variant of LTL called LocTL. 
  \end{enumerate}
 \end{remark}

\noindent
We end this section by introducing our model for distributed automata.
An \emph{asynchronous automaton}
$\Aa=\tuple{(S_\alp)_{\alp\in \Proc},s_{in},(\delta_a)_{a\in\SS}}$ is
given by 
\begin{itemize}
\item for every process $\alp$ a finite set $S_\alp$ of (local) states,
\item the initial state $s_{in} \in \prod_{\alp \in \Proc} S_\alp$, 
\item for every action $a \in\SS$ a transition relation $\delta_a
  \subseteq (\prod_{\alp\in \dom(a)}S_\alp)^2$ on tuples of states of
  processes in $\dom(a)$.
\end{itemize}

For convenience, we abbreviate a tuple $(s_\alp)_{\alp \in P}$ of local
states by  $s_P$,  where $P \subseteq \Proc$. We also  denote $\prod_{\alp \in \Proc} S_\alp$ as
\emph{global states} and $\prod_{\alp \in P} S_\alp$ as $S_P$. % Actions from $\S_p=\set{a \in\S \mid p
% \in\loc(a)}$ are denoted as \emph{$p$-actions}.

An asynchronous automaton can be seen as a sequential automaton with the
state set $S = \prod_{\alp\in\Proc} S_\alp$ and transitions $s \act{a} s'$ if
$(s_{\dom(a)}, s'_{\dom(a)}) \in \delta_a$, and $s_{\Proc\setminus
\dom(a)}=s'_{\Proc\setminus \dom(a)}$. By $\Ll(\Aa)$ we denote the set of
words labeling runs of this sequential automaton that start from the
initial state. It can be easily noted that $\Ll(\Aa)$ is
trace-closed. The automaton is \emph{deterministic} if each $\delta_a$
is a (partial) function.

\begin{figure}[htbf]
\centering
\includegraphics[scale=.5]{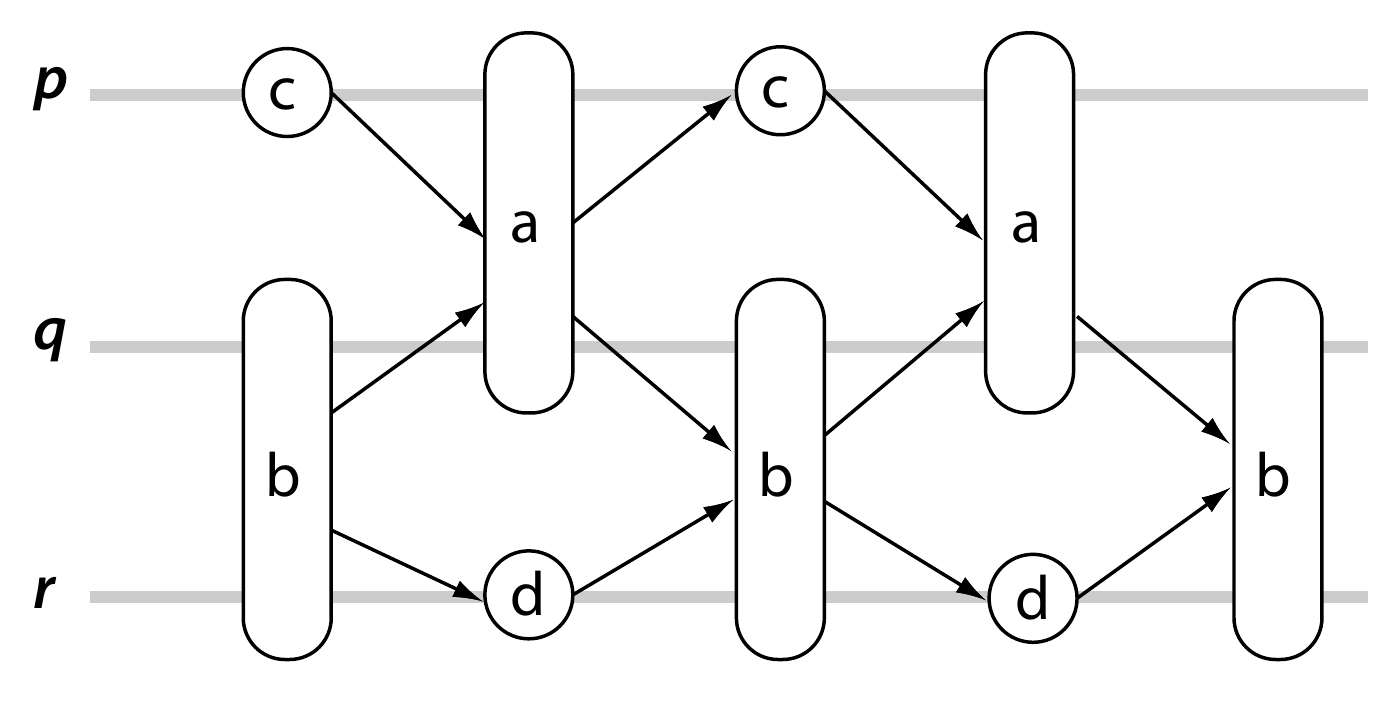}
   \caption{The pomset associated with the trace
     $t=[c\,b\,a\,d\,c\,b\,a\,d\,b]$,  
   with $\dom(a)=\{p,q\}$, $\dom(b)=\{q,r\}$, $\dom(c)=\{p\}$,
   $\dom(d)=\{r\}$. }   
  \label{fig:trace}
\end{figure} 

\begin{example}
  Let us consider the asynchronous automaton $\Aa$ given by $S_p=\{0\}$,
  $S_q=S_r=\{0,1\}$, and transition function $\delta_a(s_p,s_q)=(s_p,\neg
  s_q)$ if $s_q=1$ (undefined otherwise), $\delta_d(s_r)=\neg s_r$ if
  $s_r=1$ (undefined otherwise), $\delta_b(s_q,s_r)=(1,1)$ if $s_q \wedge
  s_r=0$  (undefined otherwise) and $\delta_c(s_p)=s_p$. Starting with
  $s_0=(0,0,0)$, an accepting run of $\Aa$ checks that between any two
  successive $b$-events, there is either an $a$ or a $d$ (or both),
  and there is a $b$-event before all $a$ and $d$.
\end{example}

Since the notion of a trace was formulated without a reference to an
accepting device, it is natural to ask if the model of asynchronous
automata is powerful enough for capturing the notion of
regularity. Zielonka's theorem below says that this is indeed the case,
hence these automata are a right model for the simple view of
concurrency captured by Mazurkiewicz traces.

\begin{theorem}\cite{zie87}\label{th:zielonka}
  Let $\dom :\SS\to(2^\Proc\setminus\oneset{\es})$ be a distribution of
  letters. %, and $I$  the induced independence relation. 
  If a language $L\subseteq \SS^*$ is regular and trace-closed then there is a
  deterministic asynchronous automaton accepting $L$ (of size exponential
  in the number of processes and polynomial in the size of the minimal
  automaton for $L$, see~\cite{ggmw10}). 
\end{theorem}

\section{Safety languages}\label{sec:basics}

A set of traces $C \subseteq \R$ is called \emph{coherent} if $C \subseteq
\pref(t)$ for some $t \in \R$. This means that $\sqcup C \in \R$
exists, and it is a prefix of $t$. By $L^c$ we denote the complement
$\R\sm L$ of $L$. Recall that $\P(L)$ is the set of prime prefixes of
traces in $L \subseteq \R$.

We use in our characterizations below a basic property of automata on traces, which is for
instance satisfied by (runs of) asynchronous automata, called \emph{forward
  diamond property}. A set $K \subseteq \SS^*$ satisfies the forward
diamond property if the following holds:
    \begin{quote}
      If $ua \in K$ and $ub \in K$, then $uab \in K$, for every $u
  \in\SS^*$ and $(a,b)\in I$.
    \end{quote}

\begin{lemma}\label{lem:a}
For $L \subseteq \R$ we have 
\[\ov L = \set{\sqcup C}{C\sse \P(L) \text{ and }   C \text{ is
    coherent}} \,.
\]
We have $\ov L = \ov K$ \IFF  $\P( L) = \P(K)$. 

\end{lemma}

\begin{proof}
  Let $X= \set{\sqcup C}{C\sse \P(L) \text{ and } C \text{ is
      coherent}}$. By definition, $X^c=U\R$ with $U=\P\sm \P(L)$, thus
  $X$ is prime-closed (and contains $L$). Let $K \supseteq L$ be
  prime-closed, thus $K^c=V\R$ with $V \subseteq \P$. Consider some
  coherent set $C \subseteq \P(L)$, and assume that $\sqcup C \in v\R$
  for some $v \in V$. But then $v \in \P(L)$, thus $K^c \cap L \not=
  \emptyset$, a contradiction. So $X \subseteq K$, which shows
  that $\ov L =X$.
\end{proof}

\begin{lemma}\label{l:regclosure}
  If $L \subseteq \R$ is recognizable, then
the prime closure  $\ov L$ is recognizable,
  too. Moreover, on input $(\SS,\dom)$ and (sequential) B\"uchi
  automaton $\Bb$ such that $L=\Ll(\Bb)$ is trace-closed, we can
  compute an exponential-size, deterministic asynchronous automaton $\Aa$
  accepting $\ov L$, such that all states of $\Aa$ are final.
\end{lemma}

\begin{proof}
  Given $L \subseteq \R$ recognizable, we have that $\P(L)$ is
  recognizable, too. Then it is easy to see that $\ov L$ is
  recognizable, by using for instance monadic second-order logic over
  traces.%  : this can be seen easily using for instance logics, but
%   we give an automata-based construction below.  Since
%   $K=\Lin(\ov L \cap \M)$ is a prefix-closed trace language
%   satisfying the forward diamond property, there exists a
%   deterministic asynchronous automaton $\Aa$ accepting it, such that all
%   reachable states are final~\cite{SEM03}. On finite and infinite traces, the
%   automaton $\Aa$ accepts precisely $\ov L$.

  Let us consider the complexity of the construction of a
  deterministic asynchronous automaton for $\ov L$ in more detail. We
  assume that the input $L$
  is given by a (sequential) B\"uchi automaton $\Bb$.  We first
  determinize $\Bb$ and get a deterministic (say Rabin) automaton $\Bb'$ for
  $L$. From $\Bb'$ we can easily construct a DFA accepting $\P(L)$: we
  just need to store the set of maximal processes in the control
  state. The resulting DFA is exponential in both $\Bb$ and
  $\Proc$. By applying the construction cited in~\refthm{th:zielonka} we
  obtain a deterministic asynchronous automaton $\Aa$ for $\P(L)$
  which is still exponential in $\Bb$ and
  $\Proc$.  Using classical timestamping we
  may assume that each local state reached by the maximal
  processes of a prime trace contains the complete information about
  the global state of $\Aa$ reached on that prime trace - the size of
  the deterministic asynchronous automaton $\Aa'$ thus obtained
  remains exponential. It remains to construct the automaton accepting
  $\ov L$. Recall that $\ov L$ contains precisely those traces where
  all prime prefixes belong to $\P(L)$. Thus, it suffices to take
  $\Aa'$ and forbid transitions that produce bad local states of
  $\Aa'$, that is, local states that are non-final viewed as global states
  of $\Aa$.  On finite or infinite traces, the automaton $\Aa'$
  accepts precisely $\ov L$. By construction, all its reachable states
are final.

\end{proof}

% \begin{proof}In the following $p$ denotes primes in $\P$. 

% We have ${\ov L}^c= \bigcup\set{p\R}{p \R \cap L = \es}= 
% \bigcup\set{p\R}{p \notin \P(L)}.$
% Now, consider $x \in \R$. Then we have $x \notin \set{\sqcup C}{C\sse \P(L) \wedge  C \text{ is coherent}}
% $ \IFF $x \in p\R$ for some $p \notin \P(L)$. 
% \end{proof}

\begin{proposition}\label{p:locsafety}
  The following are equivalent characterizations for $L \subseteq \R$:
  \begin{enumerate}
  \item $L$ is a locally-safety language.
  \item $K=\Lin(L) \subseteq \SS^\infty$ is a regular, prefix-closed
    language such that $K \cap \SS^\omega$ is a safety language, and
    $K \cap \SS^*$ satisfies the \emph{forward diamond} condition.
\item $L$ is accepted by a deterministic asynchronous automaton where
  all reachable states are final.
  \end{enumerate}
\end{proposition}

\begin{proof}
  The implications $(1) \Rightarrow (2)$ and $(3) \Rightarrow (1)$ are
  immediate. For $(2) \Rightarrow (3)$ let us assume that $K=\Lin(L)$
  is regular, prefix-closed and satisfies the two additional
  conditions in the statement.
%  We obtain
%   $K=\Ll(\Aa)$ for some deterministic automaton $\Aa$ with all states
%   final. Notice that
  Since $K \cap \SS^*$ is prefix-closed, trace-closed and
  satisfies the forward diamond property, there exists a
  deterministic asynchronous automaton 
  $\Bb$ recognizing $K \cap \SS^*$ (equivalently, the set of finite
  traces in $L$) such that all
  reachable states are final~\cite{SEM03}. Since $K$ is assumed to be
  prefix-closed and $K \cap \SS^\omega$ is a safety language, we
  obtain that  the automaton $\Bb$ accepts precisely $\ov L=L$ over $\R$.
\end{proof}

\begin{example}
  Assume that $\SS=\oneset{a,b,c}$ with $\dom(a)=\{\alp\}$,
  $\dom(b)=\{\beta\}$ and $\dom(c)=\{\alp,\beta\}$. The trace language ``no two
consecutive $c$'s'' is a locally safety language, and it can be
recognized by an asynchronous automaton where both processes remember
their last action, and do not allow two consecutive $c$'s. 

The trace language ``no $a$ in parallel with a $b$'' is not a locally
safety language (but it is Scott-closed).
\end{example}

For first-order languages we have, as usual, also a characterization
by temporal logics:

\begin{proposition}\label{p:folocsafety}
  The following are equivalent characterizations for $L \subseteq \R$:
  \begin{enumerate}
  \item $L$ is a locally-safety language definable in first-order
    logic. 
\item $L$ is definable by a globally past formula in LocTL.
  \item $K=\Lin(L) \subseteq \SS^\infty$ is a first-order, prefix-closed
    language such that $K \cap \SS^\omega$ is a safety language, and
    $K \cap \SS^*$ satisfies the \emph{forward diamond} property.
  \end{enumerate}
\end{proposition}

\begin{proof}
  The equivalence $(1) \Leftrightarrow (2)$ follows from~\cite{dg09thiagu},
  and the implication $(1) \Rightarrow (3)$ is immediate. For $(3)
  \Rightarrow (1)$ it suffices to show that $L=\ov L$ (since we know
  by~\cite{ebimus96} that $L$ must be first-order). So let $t=\sqcup C$, with $C
  \sse \P(L)$ coherent. For every $u \in \P(L)$ and every
  linearization $x$ of $u$, we have $x \in K$ since $K$ is
  prefix-closed. Moreover, if $\oneset{t,t'}$ is coherent and $K$
  contains all linearizations of $t$ and $t'$, respectively, then by
  the forward diamond property, $K$ contains some (and thus all)
  linearization(s) of $t \sqcup t'$. This shows the claim for finite
  traces $t$. For infinite traces it follows from $K \cap \SS^\omega$
  being a safety language.
\end{proof}

\section{Local monitoring}\label{sec:locmon}

Here and in the following we write $s \leq L$ for a (finite) trace $s \in \R$
and a language $L \subseteq \R$ if there exists some $t \in L$ with $s
\leq t$. 

\begin{definition}\label{df:locmon}
 A set $L \subseteq \R$ is called \emph{\locmon} if for all $s\ \in \P$
 there exists some $t \in \P$ with (1) $s \le t\R$ and (2) either $t\R
 \sse L$ or $t\R \sse L^c$.
\end{definition}

Notice that in the definition of \locmon sets, the first condition
says that $\oneset{s,t}$ is coherent. So a set $L$ is \emph{\locmon} if for
every \emph{prime} trace $s$ there is another \emph{prime} trace $t$ that is
coherent with $s$ and such that after $t$ we know that every 
extension belongs either to $L$ or to its complement $L^c$.

The following lemma extends a well-known observation from words to
traces: 

\begin{lemma}\label{lem:b}
  Every prime-closed trace language is \locmon. In particular, every
  locally-safety (or locally-co-safety) language is \locmon.
\end{lemma}

\begin{proof}
Let $L = \ov L$ and $s \in \P$. If $s\R$ is not a subset of $L$, then 
there exists some $t=sx \in L^c$. Since $L$ is prime-closed this means
that there is some $u\in \P\sm \P(L)$ with $u\leq t$.  
But then $\oneset{u,s}$ is coherent, thus $s\leq u \R$ and $u \R \sse L^c.$
\end{proof}

The next proposition characterizes \locmon sets in terms of the
closure operator defined in the previous section: 

\begin{proposition}\label{prop:c}
$L \subseteq \R$ is \locmon \IFF $\ov L \cap \ov{L^c}$ does not contain any
non-empty prime-open subset.  
\end{proposition}

\begin{proof}
First, assume by contradiction that $L$ is \locmon, but  $s \R \sse
\ov L \cap \ov{L^c}$ 
for some $s\in \P$. By symmetry in $L$ and $L^c$ we may assume that we
find $t \in \P$ and $s \leq t\R \sse L$. Hence, $t \notin \P(L^c)$ and
thus $t\R \cap \ov{L^c} = \es$.
But $s\R \cap t\R \neq \es$. Contradiction. 

For the other direction let $s \in \P$. We may assume (again by
symmetry in $L$ and $L^c$) that $s\R \cap {\ov L}^c \neq \es$. Hence,
there is $x\notin \ov L$ with $s \leq x$. This implies that there is
$t\in \P\sm \P(L)$ with 
$s \leq t\R$. Thus, $t\R \sse L^c$ and $L$ is \locmon. 
\end{proof}

We state now the main result of this section, which shows that
whenever a recognizable property over traces is \locmon, we can build
a monitor that is of the same type as the system on which it runs,
i.e., an asynchronous automaton.

\begin{theorem}\label{th:locmon}
Let $L \sse \R$ be recognizable. Then we can decide whether $L$ is \locmon.
Moreover, if $L$ is \locmon, then we find a deterministic asynchronous 
finite state monitor for $L$. 
\end{theorem}

\begin{proof}
  By Lemma~\ref{l:regclosure} there exist deterministic asynchronous
  automata $\Aa$, $\Aa'$ accepting $\ov L$ and $\ov{L^c}$, resp., such
  that all their reachable states are final. 

  Let $(\de_a)_{a \in\SS},(\de'_a)_{a \in \SS}$ be the transition
  functions of $\Aa,\Aa'$, resp. We modify the product automaton $\Aa
  \times \Aa'$ to a (deterministic) asynchronous automaton $\Cc$ with
  transition functions $(\D_a)_{a \in \SS}$: first we add two local
  states $\bot_\alp$, $\top_\alp$ on each process $\alp\in
  \Proc$. Consider $a \in\Sigma$ and some trace $t$ on which $\Aa$
  reaches state $s$ and $\Aa'$ reaches state $s'$. Note that $ta$ belongs to
  one of $\ov L$ or $\ov{L^c}$ (or both). If $\Aa$ has no
  $a$-transition on $s_{\dom(a)}$ then we add
  $\D_a((s_\alp,s'_\alp)_{\alp\in\dom(a)})= (\bot_\alp)_{\alp \in
  \dom(a)}$. If $\Aa'$ has no $a$-transition on $s'_{\dom(a)}$ then we
add the transition $\D_a((s_\alp,s'_\alp)_{\alp\in\dom(a)})=
(\top_\alp)_{\alp \in \dom(a)}$. The first case corresponds to $ta\R
\cap \ov L=\es$, the second one to $ta\R \cap \ov{L^c}=\es$. Else,
$\D_a((s_\alp,s'_\alp)_{\dom(a)})$ is defined as the componentwise
product of $\de_a(s_{\dom(a)})$ and $\de'_a(s'_{\dom(a)})$.
Finally, for each $a \in \SS$ and each tuple
$\hat{s}_{\dom(a)}$ of states of $\Aa \times \Aa'$: if some component
of $\hat{s}_{\dom(a)}$ is $\bot$, then all components of
$\D_a(\hat{s}_{\dom(a)})$ become $\bot$, and symmetrically for $\top$. The
language $L$ is not \locmon \IFF the automaton $\Cc$ has some infinite
run where no process gets into state $\top$ or $\bot$.
\end{proof}

\begin{proposition}\label{prop:lowb}
The following problem is $\PSPACE$-hard:
\begin{itemize}
\item Input: A B\"uchi automaton $\Bb = \tuple{Q, \Sigma, \del, q_0, F}$.
 \item Question: Is the accepted language $\Ll(\Bb) \sse \Sig^\oo$ \moni?
\end{itemize}
\end{proposition}

\begin{proof}The universality problem for non-deterministic finite
  automata (NFA) is one of the well-known \PSPACE complete
  problems. We reduce this problem to the problem of monitorability.
 
Start with an NFA
$\Aa = \tuple{Q', \Gam, \del', q_0, F'}$. We will construct a B\"uchi automaton $\Bb$ 
such that we have $\Ll(\Aa)= \Gam^*$ \IFF $\Ll(\Bb) \sse \Sig^\oo$ is \moni. 

For this we use a new letter $b$ and we let $\Sig = \Gam \cup \os{b}$. 
We use three new states $d,e,f$ and we let $Q = Q'\cup \os{d,e,f}$.
The repeated (or final) states of $B$ are defined as 
$F= \os{e,f}$. The initial state is the same as before: $q_0$. 
It remains to define $\del$. We keep all arcs from $\del'$ and we add the 
following new arcs. 
\begin{itemize}
\item $ q \act{b} d \act{a} e \act{a} e$ for all $q \in Q' \sm F'$ and
  all $a \in \Gam$.  
 \item $ e \act{b} d \act{b} d$
 \item $ q \act{b} f \act{c} f $ for all $q \in F'$ and 
 all $c \in \Sig$. 
 \end{itemize}
 In order to understand the construction, consider what happens if we
 reach state $d$ or state $f$. Starting in $f$ we accept everything,
 because we loop in a final state of $\Bb$. On the other hand starting
 in $d$ we accept all words except those which end in
 $b^\oo$. Starting in $d$ we are nowhere \moni.
 
 Now, let $w\in \Sig^*$. This can be written as $uv$ where $u\in
 \Gam^*$ is the maximal prefix without any occurrence of $b$.  
 
 Assume we have 
 $\Ll(\Aa)= \Gam^*$, then there is path from $q_0$ to $f$ labelled by
 $wb$ since  
 reading $u$ leads us to some state in $F'$. This implies that 
 $wb\Sig^\oo \sse \Ll(\Bb)$ for all $w\in \Gam^*$; and $\Ll(\Bb)$ is \moni. 
 
 On the other hand, if $\Ll(\Aa) \neq \Gam^*$, then there is some word
 $u \in \Gam^*$ such that $u$ leads to states in $Q'\sm F'$, only. Thus, reading
 $ub$ we are necessarily in state $d$. The language $\Ll(\Bb)$ is not
 \moni, due to the word $ub \in \Sig^*$.
\end{proof}

We have a matching upper bound for Büchi automata in the theorem
below. Note that the input is a B\"uchi automaton accepting a
trace-closed language, therefore we may see the accepted language also
as a subset of $\R$.

\begin{theorem}\label{thm:compl}
%Let $(\Sig, D)$ be connected. \vdd{necessary?}
  The following problem is $\PSPACE$-complete:
\begin{itemize}
\item Input: A B\"uchi automaton $\Bb = \tuple{Q, \Sigma, \del, q_0, F}$
  and $(\SS,\dom)$ such that $\Ll(\Bb)$ is trace-closed.
 \item Question: Is the accepted language $\Ll(\Bb) \sse \R$ \locmon?
%  \vdd{Hier und unten im Beweis unterscheiden wir nicht zwischen Wortsprachen 
%  und Spursprachen, dies sollte irgendwo explizit stehen, oder?}
% \anca{siehe vor dem Theorem\vdd{OK}}
\end{itemize}
\end{theorem}

\begin{proof}
 For a subset $P \sse Q$ let us write $\Ll(\Bb,P)$ for the accepted language of 
 $\Bb$ when $P$ is used as a set of initial states. We say that $P$ is \emph{good} 
 if either $\Ll(\Bb,P) = \Sig^\oo$ or $\Ll(\Bb,P) = \es$. The predicate whether 
 $P$ is good can be computed in \PSPACE. For a letter $a \in \Sig$ and $P,P' \sse Q$ 
 we define another predicate 
 $\Reach(P,P',a)$, which is defined to be \emph{true}, if:
 \[P' = \set{q\in Q}{\exists p\in P \; \exists ta \in \P \;
   \text{ and } p \act{ta} q}\,.\]
\noindent
Note that $\Reach(P,P',a)$ is computable
in \PSPACE, too.
If there is no $a \in \Sig$ such that $\Reach(\os{q_0},P',a)$ becomes
true for some good $P'\sse Q$, then $L= \Ll(\Bb)$ is not
\locmon. Thus, we may assume that such $P$ and $a$ exist. If there are
two letters $a$ and $b$ in different connected components of $(\Sig,
\dom)$ with this property, then $L$ is \locmon.  Hence we assume in
the following that there is only one component where such a letter $a$
exist. Indeed, letters occurring in some prime traces belong to a
single connected component of $(\Sig, \dom)$; and due to
$\Reach(\os{q_0},P',a)$ it is enough to consider monitorability of
prime traces which belong to the same component as the letter
$a$. Since every such prime trace can be made longer such that it ends with
this letter $a$, we fix $a$ in the following.

Now, the language $L \sse \R$ is \locmon \IFF for all $P\sse Q$ such
that $\Reach(\os{q_0},P,a)$ holds, % for some $a\in \Sig$,
there is some good subset $P'$ such that we have $\Reach(P,P',a)$.

To see this, let $L \sse \R$ be \locmon. Consider a subset $P$ such that 
$\Reach(\os{q_0},P,a)$ holds.  %for some $a\in \Sig$. 
This corresponds to some
word $s$ such that the corresponding trace $s= s'a$ is a prime. 
Since $L$ is \locmon, %and $(\Sig,D)$ is connected,  
there exists some prime $t$  such that $s \leq t\R$ and 
either $t \R \sse L$ or $t \R \sse L^c$. However, by the assumption above, 
we may assume that $s$ and $t$ belong to the same component. We can make 
$t$ longer and actually assume $s \leq t$ and such that $t=t'a$. 
Choose some representing word  $w$ for $t$. If $P'$ is the subset of states we can reach after reading 
$w$ starting in $q_0$ we have  $\Reach(P,P',a)$.
The set $P'$ is good, because $L$ is trace-closed. 
Indeed, if $t \R \sse L$, then $w\Sig^\oo\sse L$, hence
$\Ll(\Bb,P') = \Sig^\oo$. If  
$t \R \sse L^c$, then  $\Ll(\Bb,P') = \es$.  

For the converse it is clear that the condition is strong enough to
ensure local monitorability of $L$.  
\end{proof}

The condition to monitor a single language might be an unnecessary restriction.  We can imagine a certain family of properties or languages 
$L_1, \ldots, L_n$ and we content ourselves with a monitor which selects 
one of these possibilities, even if certain $L_i$ and $L_j$ do intersect
non-trivially for $i \neq j$. This leads to the following definition. 

\begin{definition}\label{def:fam}
Let $n\in \N$ and $L_1$, \ldots, $L_n$ be subsets of $\R$. 
We say that the family  $\os{L_1, \ldots, L_n}$ is \locmon,
if
$$\forall s \in \P \;\exists t \in \P\; \exists 1\leq i \leq n: \;  s \leq t\R \sse L_i.$$ 
\end{definition}

\begin{remark}\label{rem:e}
A language $L$ is \locmon \IFF the family $\os{L, L^c}$ is \locmon. 
\end{remark}

A distributed alphabet $(\SS,\dom)$ can be split into 
several  \emph{connected components}. This is a partition 
$\SS=\SS_1 \cup \cdots \cup \SS_k$ such that all $\SS_i$ are non-empty and  $\SS_i \times \SS_j \sse I$
for all $1 \leq i < j \leq k$. We say that $(\SS,\dom)$ is \emph{connected}, 
if $k = 1$ and \emph{disconnected} otherwise. For $k \geq 2$ we can write
$\R = \R' \times \R''$ such that $\R'$ and $\R''$ are both infinite. 

\subsection{Disconnected case}\label{ssec:discon}
We assume in this section that $(\SS,\dom)$ is disconnected 
and we write $\R = \R' \times \R''$.
Let $L \sse \R$. If $L$ is \locmon then, necessarily $s\R \sse L$ or $s\R \sse L^c$ for some prime $s \in \P = \P(\R') \cup \P(\R'')$. By symmetry we may assume 
$s \in \P(\R')$ and $s\R \sse L$. As a consequence, 
there is no $t \in \P(\R'')$ such  $t\R \sse L^c$. On the other hand, 
if there is some prime $t \in \P(\R'')$ such  $t\R \sse L$, then $L$ is \locmon for a trivial reason: For every prime trace $u \in \P$ we either have 
$u \in \R'$ or $u \in \R''$; and by choosing either the prime 
$s$ or $t$ in the other component as $u$ we satisfy the required condition for 
$L$ to be \locmon. 

Hence we are only interested in the case that there is no prime $t \in
\R''$ such that $t\R \sse L$. In this case we can reduce the problem
whether $L$ is \locmon to the component of $\R'$ as follows: First,
let us define languages of prime traces $L_1 = \set{u\in
  \P(\R')}{u \R \sse L}$ and $L_2 = \set{u\in \P(\R')}{u \R \sse
  L^c}$.  Note that if $L$ is recognizable, then $ L_1, L_2$, as well
as $ L_1\R', L_2\R'$,
 are recognizable too. 
Moreover, we can construct the corresponding automata. 

\begin{theorem}\label{thm:discon}
  Let $L \sse \R = \R' \times \R''$ and assume that there is some $s
  \in \P(\R')$ such that $s\R \sse L$ but there is no $t \in \P(\R'')$
  with $t\R \sse L$.  Then $ L$ is \locmon \IFF the family 
$\os{L_1\R', L_2\R'}$ is \locmon w.r.t.~$\R'$. 
\end{theorem}

\begin{proof}
  First, let $L$ be \locmon and $s \in \P$ be a prime. Choose some
  prime $t\in \P$ with $s \leq t\R$ such that either $t\R \sse L$ or
  $t\R \sse L^c$.  We cannot have $t \in R''$, hence $t \in
  \P(R')$. Thus, either $t \in L_1$ or $t \in L_2$. It follows that
  $t\R' \sse L_1\R'$ or $t\R' \sse L_2\R'$, and hence $\os{L_1\R',
    L_2\R'}$ is \locmon w.r.t.~$\R'$.

  For the other direction let $\os{L_1\R', L_2\R'}$ be \locmon
  w.r.t.~$\R'$. Then for every prime $u \in \P(\R')$ there is some $v
  \in \P(\R')$ such that $u \leq v\R'$ such that either $v\R' \sse
  L_1\R'$ or $v\R' \sse L_2\R'$. In particular, either $v\in L_1$ or
  $v\in L_2$, since $r \leq v$ with $r \in L_i$ implies $v \in
    L_i$. By definition, either $v \R \sse L$ or $v \R \sse
  L^c$. Thus, $L$ is \locmon on all primes of $\R'$. Now, let $u \in
  \P(\R'')$. By assumption there is some $s \in \P(\R')$ such that
  $s\R \sse L$. Since $\R = \R' \times \R''$ we have $u \leq
  s\R$. Thus, $L$ is \locmon.
\end{proof}

\subsection{Connected case}\label{ssec:con}
Recall that a distributed alphabet $(\SS,\dom)$ is \emph{connected} if it cannot
be partitioned as $\SS=\SS_1 \cup \SS_2$ such that $\SS_1 \times \SS_2
\subseteq I$ with $\SS_1 \not=\es\not=\SS_2$. For connected
$(\SS,\dom)$ we obtain a nicer characterization of \locmon sets:

\begin{lemma}\label{lem:concase}
Let $(\Sigma, \dom)$ be connected. Then 
$L$ is \locmon \IFF 
$$\forall s \in \P \; \exists s\leq  t \in \P: \;  t\R \sse L \vee  t\R \sse L^c.$$
\end{lemma}

\begin{proof}
  Let $L$ be such that $\forall s \in \P\; \exists t \in \P: \; s \leq
  t\R \sse L \vee s \leq t\R \sse L^c.$ We have to show that we can
  choose $s$ to be a prefix of $t$. But this is clear: if $s \leq
  t\R$, then there is a prime $p$ with $s\leq p$ and $t \leq p$. The
  result follows because $p\R \sse t\R$ in this case.
\end{proof}

\begin{proposition}\label{prop:bool}
The following assertions are equivalent.
\begin{enumerate}
\item $(\Sigma, \dom)$ is connected.
\item The family of \locmon sets is closed under finite union. 
\item The family of \locmon sets is a Boolean algebra. 
\end{enumerate}
\end{proposition}

\begin{proof} Since the \locmon property is symmetric for $L,L^c$, the
  last two items of the proposition are equivalent. 
Let $(\Sigma, \dom)$ be connected, we show that \locmon is preserved by
taking finite unions. Let $L$ and $K$ be
  \locmon and consider $s\in
  \P$. If we find $s \leq t \in \P$ and either $t\R \sse L $ or $t\R
  \sse K$, we are done.  Hence there is $s \leq t\in \P$ and $t\R \sse
  L^c $. Now, we may assume that there is $t \leq u\in \P$ and $u\R
  \sse K^c$. But then $s \leq u$ and $u\R \sse (L\cup K)^c$. 

 Conversely, let $a, b \in \Sigma$ be in different connected components
  of $(\Sigma, \dom)$ and let $L=$ ``no occurrence of $a$'' and $K=$ ``no
  occurrence of $b$''. Both sets are \locmon, since they are
  prime-closed. %  Indeed, let $s\in \P$
%   and $p= ta\in \P$ be a the prime prefix of $sa$ ending in the last
%   $a$. Then we have $s \leq p\R$ and $p\R \cap L = \es$.
  However, for
  every prime $s$ we have $s\in L\cup K$ and $s\R \cap (L\cup K)^c
  \neq \es$. This shows that $L\cup K$ is not \locmon.
 \end{proof}

Again, for connected alphabets and a family of languages, we can make
the condition to be \locmon more precise by using \reflem{lem:concase}. Indeed, 
if $(\Sigma, \dom)$ is connected, then a family $\os{L_1, \ldots, L_n}$ is \locmon
\IFF 
$$\forall s \in \P \; \exists\; s\leq  t \in \P \; \exists 1\leq i
\leq n: \;  t\R \sse L_i.$$

% 
%\begin{definition}\label{def:locsep}
%Let $L,K\sse \R$ be disjoint sets. 
%We say that a local monitor may discard $L$ or $K$, if $\forall s \in \P \exists s \leq t \in \P : t\R \cap L = \es \vee t\R \cap K = \es$. 
%\end{definition}

%\begin{proposition}\label{prop:c}
%Let $(\Sig, D)$ be connected. 
%Let $L,K$ be two disjoint prime-$\del$ sets, then some 
%local monitor may discard  $L$ or $K$. 
%\end{proposition}

\begin{theorem}\label{thm:nix}
Let $(\Sig, \dom)$ be connected, % \vdd{unklar, ob notwendig}
and $L_1$, \ldots, $L_n$ be subsets of $\R$ such that 
\begin{enumerate}
\item $\R= L_1 \cup \cdots \cup L_n$.
\item Each $L_k$ is a countable union of prime-closed sets. 
\end{enumerate}
Then the family $\os{L_1, \ldots, L_n}$ is \locmon.
\end{theorem}

\begin{proof}
  We give the proof for $n=2$, the one for $n>2$ is similar. Let
  $L=L_1^c$ and $K=L_2^c$. Write $L = \bigcap_{i \ge 0} U_i\R$
  and $K = \bigcap_{i \ge 0} V_i\R$ where all $U_i,V_i \subseteq
  \P$.  Without restriction we have $U_0\R = V_0 \R=\R$. 

  By contradiction, assume that $\oneset{L_1,L_2}$ is not
  \locmon. This means that we can find some $s \in\P$ such that for
  all $t \in\P$ with $\oneset{s,t}$ coherent it holds that $t\R \cap L
  \neq \es \neq t\R \cap K$.  Let $p_0= x_0 = q_0 = y_0 = s$.

 By induction let for some $k\geq 1$ prime traces $p_i$, $x_i$, $q_i$, and
 $y_i$ for all $0 \leq i < k$ be defined such that $U_i \ni p_i\leq
 x_i\leq y_i$, $V_i \ni q_i \leq y_i$, and $y_{i-1} \leq x_i$.
 
 We define $x_k,p_k$ as follows. Since $s \leq y_{k-1}\in \P$ we have
 by assumption $y_{k-1}\R \cap L \not=\es$, and thus we find
 $y_{k-1}\leq x \in L$. Thus, there is $p_k \in U_k$ with $p_k \leq
 x$. The set $\smallset{y_{k-1},p_k}$ is coherent, hence there is
 common finite trace $w$ with $y_{k-1}\leq w$ and $p_k\leq w$. Since
 $(\Sig, \dom)$ is connected, we find some prime $x_k\in \P$ with $w
 \leq x_k$. The definition of $y_k$ follows the same pattern.  We have
 $s \leq x_1 \leq y_1 \leq x_2 \cdots$ and $x = \sqcup\set{x_i}{i\in
   \N}$ exists.  However, $x\in \bigcap_{i \ge 0} U_i\R \cap
 \bigcap_{i \ge 0} V_i\R$. Contradiction, because $L \cap K = \es$.
\end{proof}

\begin{remark} 
  Notice that the above proof still works if $(\SS,\dom)$ has only two
  connected components. In the general case it is open whether the
  statement of \refthm{thm:nix} still holds.
\end{remark}

% \anca{\refthm{thm:nix} still holds for only 2 conn. comp.}

% \vdd{Wollen wir dies als Bemerkung aufnehmen und sagen, dass der allgemeine Fall offen ist?}

% \begin{remark}\label{rem:warumcon}
% Stimmt in der neuen Def nicht. Unklar
% \refthm{thm:nix} does not hold in general, if $(\Sig, D)$ is disconnected.
% Indeed, let $(a,c), (b,c) \in I$, $(a,b) \in D$: Then $L = ac^\oo\R$  and $K = bc^\oo\R$ are disjoint countable intersections of prime-open sets. 
% Choose $s = c$, then for all $s \leq t \in \P$ we have $t \in c^+$ and $L \cup K \sse t\R$.
% \end{remark}

\section{Infinite traces}\label{sec:inf}
Prime-closed languages are prefix closed, so they always intersect. In
particular, for any language $L \subseteq \R$, it can never happen
that both $L$ and $L^c$ are countable unions of prime-closed sets (or
equivalently, countable intersections of prime-open sets), as required by
\refthm{thm:nix}. 

Thus, in order to define an trace analogue of $G_\delta \cap F_\sigma$
we will restrict our attention to infinite traces where a (given)
subset $\GG$ of processes is active infinitely often and ``sees'' all
other processes. In this way monitoring can be performed by processes in $\GG$.   Another motivation for the new notion is due to
the fact that in order to monitor a language we should be able to
gather information into longer and longer prime prefixes. 

For a  finite trace $t$ we write $\max(t) \subseteq \GG$ if $\dom(a)\cap\GG
    \not=\es$ for each $a \in \max(t)$.

\begin{definition}\label{def:inf}
  Let $\GG$ be a (non-empty) subset of $\Proc$. A trace $x$ is called
  \emph{$\GG$-infinite} if 
  \begin{itemize}
  \item Every process from $\GG$ has infinitely many actions in $x$.
  \item $x$ can be written as $x=x_0 x_1 \cdots$ such that $\max(x_n)
    \subseteq \GG$ for each $n\ge 0$.
\item $\alf(x)$ is connected.
  \end{itemize}

The set of $\GG$-infinite traces is written as $\R_\GG$.%  and we denote
% traces from $\R_\GG$ as \emph{strictly infinite traces}.
%  A trace $x$ is called \emph{strictly infinite}, if $1 \neq x$,
%   $\max(x) = \es$ and . Thus, $x$ is a non-empty (infinite) trace without
%   maximal elements. The set of strictly infinite traces is denoted by
%   $\R^{\inf}$.
\end{definition}

\begin{remark}
  If $\GG$ is a singleton, then for every trace $x \in \R_\GG$, both $\alf(x)$
  and $\alfinf(x)$ are connected (and non-empty).
\end{remark}

In the following everything is within $\GG$-infinite traces, for a
fixed set $\GG\subseteq\Proc$.
In particular, the notion of closed and open are meant to be
induced. The notion of \locmon is also relative to $\R_\GG$: a set $L
\subseteq \R_\GG$ is \locmon if $\forall s \in \P(\R_\GG) \, \exists s \leq  t
\in \P(\R_\GG) : \; t\R \cap \R_\GG \sse L \vee t\R \cap \R_\GG \sse
L^c$ (where $L^c=\R_\GG\setminus L$).

\begin{definition}\label{def:}
Let $\GG \subseteq \Proc$ be a non-empty set of processes.
\begin{enumerate}
% \item $L \sse \R_X$ is called \emph{prime-clopen}, if
% $L$ is prime-open and prime-closed with respect 
% to $\R_X$.
\item A set $X \subseteq \R_\GG$ is prime-$G_\delta$ if it has
  the form $X=\bigcap_{i \ge 0} U_i$ where all
  $U_i$ are prime-open in $\R_\GG$. The family of  prime-$G_\delta$-sets
  is denoted $\PGd$.
\item A set $X \subseteq \R_\GG$ is prime-$F_\sigma$ if its complement is
  prime-$G_\delta$. The family of   prime-$G_\delta$-sets is
  denoted $\PFs$.
\end{enumerate}

\end{definition}

\begin{example}
  Let $\GG=\Proc=\oneset{\alpha,\beta}$ and $\SS=\oneset{a,b,d}$ with
  $\dom(a)=\{\alpha\}$, $\dom(b)=\{\beta\}$ and
  $\dom(d)=\oneset{\alpha,\beta}$. Let $L\subseteq \R_\GG$ contain all
  traces without the (trace) factor $abd$. Such traces are formed
  either by a trace from
  $((a^*+b^*)d^+)^*(a^*+b^*)d^+$ followed by $a^\oo b^\oo$, or they
  belong to $((a^*+b^*)d^+)^\oo$. Clearly, $L$ is prime-closed. The
  complement of $L$ is in $\PFs$, since $L^c=\bigcup_{w
    \in \SS^*, i,j >0} X_{w,i,j}$  where $X_{w,i,j}$ contains all traces from $\R_\GG$
  with prefix $w a^i b^j d$. Each $X_{w,i,j}$ is
  prime-closed. 
  
\end{example}
 
%% anca: does not hold for $\R^{\inf}$, and neither for $\R_X$.
% \begin{lemma}\label{lem:a}
% For $p \in \P$ we have $p\R \cap \R^{\inf}$ is prime-clopen.
% \end{lemma}

% \begin{proof}
%   The set $p\R \cap \R^{\inf}$ is prime-open in $\R^{\inf}$. Its
%   complement is $$\bigcup\set{q\R \cap \R^{\inf}}{\os{p,q} \text{ is
%       not coherent}, q\in \P}.$$
% \end{proof}

The next lemma generalizes the case of $\omega$-words. Note
that we need the restriction to $\R_\GG$ (or some similar restriction). As an example, consider
$\SS=\oneset{a,b}$ with $(a,b) \in I$. The language $L=a\R$ is
prime-open. But its complement $L^c=b^\infty$ cannot be
written as countable intersection of prime-open sets in $\R$, since we cannot
avoid occurrences of $a$ in such sets.

\begin{lemma}\label{lem:pa}
Prime-closed sets of $\R_\GG$  are in $\PGd$.
\end{lemma}

\begin{proof}
% A closed set is given by a language $L$ and by taking the all $\sqcup C$ where $C$ is coherent and $C \sse \P(L)$. For a coherent set $C$ and $k \in \N$ let $C_k = \set{p\in C}{\abs p \geq k}$. 
% The pure infinite traces of $\ov L$ are defined by the 
% intersection $\bigcap_{k\in \N} \set{\sqcup C_k}{C \text{ is coherent and } C \sse \P(L)}.$
% Let $x \in \Ri$. Then we have 
%  $x \in \bigcap_{k\in \N} \set{\sqcup C_k}{C \text{ coherent}, C \sse \P(L)}$ 
%  \IFF 
%  $x \in \bigcap_{k\in \N} \set{p\R}{p\in \P(L) \wedge \abs p \geq
%  k}.$
Let $L \subseteq \R_\GG$ be prime-closed. By definition, every $\sqcup C \in
\R_\GG$ where $C$ is coherent and $C \sse \P(L)$, belongs to $L$. For $K
\subseteq \P$, $\alpha \in\GG$ and $k\in\N$ let 
\[K_{\alpha,k} = \set{p\in K}{\abs p \geq k, \;
  \alpha\in\dom(\max(p))}.\] We claim that
\[L=\bigcap_{k\in
  \N,\, \alpha \in \GG} \P(L)_{\alpha,k}\, \R_\GG \,.
\]
The inclusion from left to right follows from $L \subseteq \R_\GG$ and
the definition of $\R_\GG$. Let $x \in \R_\GG$ be such that for every
$k\in\N$ and $\alpha\in\GG$, there is some $p_{\alpha,k} \le x$ with
$p_{\alpha,k} \in \P(L)_{\alpha,k}$. By definition of $\R_\GG$ and of
$ \P(L)_{\alpha,k}$, we have that $x=\sqcup \set{p_{\alpha,k}}{k\in
  \N,\, \alpha \in \GG}$. Hence, $x$ is of the form $\sqcup C$ for $C
\subseteq \P(L)$ coherent, and thus in $L$.
\end{proof}

\begin{theorem}\label{thm:main}
\begin{enumerate}
\item $\PGd \cap \PFs$ is a Boolean algebra containing all 
prime-open and all prime-closed subsets of $\R_\GG$. 
\item All $\PGd \cap \PFs$ subsets of  $\R_\GG$ are \locmon. 
\end{enumerate}
\end{theorem}

\begin{proof}
 $\PGd$ is closed under union. Hence, $\PGd \cap \PFs$ is a Boolean algebra.
It contains all 
prime-open and all prime-closed subsets of $\R_\GG$ by \reflem{lem:pa}. 

The proof of the second claim follows along the same lines as
the one of \refthm{thm:nix}. Assume that $\R_\GG \not=\es$ and choose some
connected subalphabet $\SS'$ of $\SS$ that contains for each
$\alpha\in\GG$ some letter $a$ with $\alpha\in\dom(a)$. The prime
traces $x_k,y_k$ can be chosen such that $\max(x_k) \subseteq \GG$,
$\max(y_k) \subseteq \GG$, and $\alf(x_k^{-1}y_k)=\alf(y_{k-1}^{-1}x_k)=\SS'$. Thus,
$x=\sqcup_{i} x_i \in \R_\GG$. 
\end{proof}

Asynchronous B\"uchi and Muller automata have been studied
in~\cite{gaspet95,dm94}. McNaughton's theorem \cite{McNau66}
stating the equivalence of 
non-deterministic B\"uchi and deterministic
Muller automata over omega-word languages, extends to recognizable languages of infinite traces
and asynchronous automata~\cite{dm94}. If we restrict to traces from $\R_\GG$,
then the B\"uchi and Muller acceptance conditions are simpler:

\begin{definition}
  Let $\GG \subseteq \Proc$ be a non-empty set of processes, and let
  $\Aa=\tuple{(S_\alp)_{\alp \in \Proc}, (\delta_a)_{a \in \SS}, s^0}$
  be an asynchronous automaton.
\begin{enumerate}
\item A B\"uchi acceptance condition is a set $F\subseteq S_\GG$.

An infinite run $s^0=s_0,a_0,s_1,a_1,\ldots$ of $\Aa$ is accepting if
for some $f_\GG \in F$ and for every $\alp\in\GG$, there are
infinitely many $n\ge 0$ with $(s_n)_\alp=f_\alp$. 
\item A Muller acceptance condition is a set $\Ff\subseteq \prod_{\alp
    \in \GG} 2^{S_\alp}$. 

An infinite run $s^0=s_0,a_0,s_1,a_1,\ldots$ of $\Aa$ is accepting if
for some $T_\GG \in \Ff$ and for every $\alp\in\GG$, the set of states
from $S_\alp$ such that $(s_n)_\alp=f_\alp$ for infinitely many $n$,
is precisely $T_\alp$. 
\end{enumerate}

\end{definition}

The language $\Ll(\Aa)$ is the set of all traces from $\R_\GG$ that
have an accepting run. 
% \vdd{Ich hatte keine Zeit, \refthm{thm:GdDetAsyn} richtig zu durchdenken, aber auch keine Fehler gefunden und ich denke, es sollte 
% OK sein. Es ist doch genau, was wir uns vorstellen, oder? 
% Ich habe keine Einwände.} 
The next result is a generalization from $\omega$-word languages to
$\R_\GG$ trace languages:

\begin{theorem}\label{thm:GdDetAsyn}
  Let $L \subseteq \R_\GG$ be recognizable. Then $L$ is in
  $\PGd$  \IFF $L$ is accepted by a deterministic B\"uchi
  asynchronous automaton.
\end{theorem}

\begin{proof}
  Assume first that $L=\Ll(\Aa)$, where $\Aa$ is a deterministic
  asynchronous B\"uchi automaton, and fix a final state $f \in F$.
  For $n >0, \alp\in\GG$ we define $K^f_{n,\alp}$ as the set of all
  traces $t \in \P$ with $\alp \in \dom(\max(t))$ and such that in the
  run of $\Aa$ on $t$, at least $n$ letters on process $\alp$ are in
  state $f_\alp$. It is easy to see that the set $\bigcup_{f \in F}\bigcap_{\alp
    \in\GG, n>0} K^f_{n,\alp}\R_\GG$ is precisely
  $\Ll(\Aa)$. The remaining of the proof will show that $\PGd$ is
  closed under finite union, thus $\Ll(\Aa) \in \PGd$.

  For the converse let $L=\bigcap_{n >0} U_n \subseteq \R_\GG$ be
  recognizable, with $U_n$ prime-open in $\R_\GG$.  We first define $V_n
  =\bigcap_{k \le n} U_n$. It is not difficult to see that each $V_n$
  can be assumed to be of the form $K_n\R_\GG$ with $\max(t) \subseteq
  \GG$ for each $t \in K_n$. Let
  now $K'_n \subseteq K_n$ consist of all elements of $K_n$ that have
  no proper prefix in $K_n$. Let $K=\bigcup_{n>0} K'_n X_n$, where
  $X_n$ is the set of traces $t$ such that (1) $\max(t) \subseteq
  \GG$, (2) $|t|_\alp\ge n$ for each $\alp\in\GG$, and (3) no proper
  prefix of $t$ satisfies (1) and (2).

  Let us first show that $L= \set{\sqcup C}{C \subseteq K, \, C
    \text{ coherent}}$. The inclusion from left to  right follows from
  $L=\bigcap_{n >0} U_n = \bigcap_{n >0}
  K_n\R_\GG= \bigcap_{n >0} K'_n\R_\GG=\bigcap_{n >0}
  K'_nX_n\R_\GG$. Conversely, let $t=x_0x_1 \ldots$ with $x_0 \cdots
  x_n \in K$ for all $n$.  Observe that we must have infinitely many $n$ such
  that $x_0 \cdots x_m \in K_n$ for some $m$, since $K'_n$ is
  prefix-free. Thus, $t \in V_n$ for infinitely many $n$ and $t \in
  U_n$ for all $n$.

  To conclude, we show that if $L= \set{\sqcup C}{C \subseteq K,
    \, C \text{ coherent}}$ for some $K$, and $L \subseteq \R_\GG$ is
  recognizable, then $L$ is the language of a deterministic
  asynchronous B\"uchi automaton. We assume as above that $\max(t)
  \subseteq \GG$ for all $t \in K$. Since $L$ is recognizable, there
  is some deterministic Muller automaton $\Aa$ with acceptance condition
  $\Ff$ and  $\Ll(\Aa)=L$. We
  may also assume that on every finite trace $t$  the states of
  processes from $\dom(\max(t))$ reached on $t$ 
  determine the states of all other processes. First we test for every
  $T \in\Ff$ if there is some trace from $\R_\GG$ accepted with
  $T$. Without restriction this is the case for all $T\in\Ff$.
  For each $T$ we can determine a reachable state 
  $s(T) \in \prod_{\alp\in\GG} T_\alp$ and finite traces $t_0(T),t(T)$
  with $\max(t_0(T), \max(t(T))
  \subseteq \GG$ such that (1)  $t_0(T)$ leads from the initial state
  to $s(T)$, (2) $t(T)$ is a loop on state $s(T)$ and (3) the
  set of $\alp$-states in the loop $t(T)$ is precisely $T_\alp$. In addition,
  $t_0(T)$ is connected.

  We claim that $\Aa$ accepts $L$ with the following (B\"uchi)
  condition: a trace is accepted if for some $T\in\Ff$, every
  state from $T_\alp$ occurs infinitely often, for every $\alp\in
  \GG$. It is clear that all of $L$ is accepted in this way by
  $\Aa$. Conversely, let $x$ be an arbitrary trace with $\max(x) \subseteq
  \GG$ and looping on state
  $s(T)$.  We have $t_0\, t(T)^\oo \in L$, so there is some $n_0$ and
  $u_0$ in $K$ such that $u_0 \le t_0 \, t(T)^{n_0}$. Since $t_0 \,t(T)^{n_0}x \,t(T)^\oo
  \in L$ we find some $n_1$ such that $u_1 \le t_0\, t(T)^{n_0}x \,t(T)^{n_1}$ for
  some $u_1 \in K$ with $u_0 <u_1$. In this way we can build a
  trace $t$ from $\R_\GG$, $t=t_0 \,t(T)^{n_0}x \, t(T)^{n_1}x \cdots$,
  with $t=\sqcup_{n\ge 0} u_n \in \set{\sqcup C}{C \subseteq K, \, C
    \text{ coherent}}$ and such that for each $\alp\in\GG$, the set of
  states from $S_\alp$ repeated infinitely often is a superset of
  $T_\alp$. The claim follows since $L=\set{\sqcup C}{C \subseteq K, \, C
    \text{ coherent}}$.

\end{proof}

\begin{remark}
  For the previous proof we do not need the connectedness
assumption in the definition of $\R_\GG$. On the other hand, it is
open whether without this assumption all
$\PGd \cap \PFs$ sets are still \locmon.
\end{remark}

\section*{Conclusion}

Our aim in this paper was to propose a reasonable notion of
distributed monitoring for asynchronous systems. We argued that 
distributed monitors should have the same structure as the system that
is monitored. We showed that properties over $\GG$-infinite
traces that are deterministic and co-deterministic, are \locmon. It
would be interesting to consider alternative restrictions to
$\GG$-infinite traces, that capture some reasonable (partial) knowledge about the
asynchronous system and for which $\PGd \cap \PFs$ sets are \locmon.

\end{document}